\documentclass[11pt]{article}
\usepackage[left=1.1in, right=1.1in, top=1in, bottom=1in]{geometry} 
\usepackage{times}
\usepackage{soul}
\usepackage{url}
\usepackage[hidelinks]{hyperref}
\usepackage[utf8]{inputenc}
\usepackage[small]{caption}
\usepackage{graphicx}
\usepackage{amsmath}
\usepackage{amsthm}
\usepackage{booktabs}
\usepackage{algorithm}
\usepackage{algorithmic}
\usepackage[switch]{lineno}
\usepackage{tikz}
\usepackage{amsfonts}
\usepackage{enumitem}

\usepackage{natbib}
\usepackage[showdeletions]{color-edits}
\addauthor{et}{blue}
\addauthor{jm}{green}
\addauthor{aa}{orange}
\addauthor{yk}{red}
\addauthor{cp}{purple}

\usepackage{listings}
\usepackage{graphicx}
\usepackage{subcaption}
\usepackage{float}

\newtheorem{theorem}{Theorem}
\newtheorem{condition}{Condition}
\newtheorem{lemma}{Lemma}

\newtheorem{definition}{Definition}

\title{Learning in Strategic Queuing Systems with Small Buffers}
\author{%
  Ariana Abel \\
  UC Berkeley \\
  \texttt{aabel14@berkeley.edu}
  \and
  Yoav Kolumbus \\
  Cornell University \\
  \texttt{yoav.kolumbus@cornell.edu}
  \and
  Jer\'onimo Mart\'in Duque \\
  Cornell University \\
  \texttt{jm2478@cornell.edu}
  \and
  Cristian Palma Foster \\
  Cornell University \\
  \texttt{cjp276@cornell.edu}
  \and
  \'Eva Tardos \\
  Cornell University \\
  \texttt{eva.tardos@cornell.edu}
}
\date{}

\begin{document}
\sloppy
\maketitle

\begin{abstract}
    We consider learning outcomes in games with carryover effects between rounds: when outcomes in the present round affect the game in the future. An important example of such systems is routers in networking, as they use simple learning algorithms to find the best way to deliver packets to their desired destination. This simple, myopic, and distributed decision process makes large queuing systems easy to operate, but at the same time, the system needs more capacity than would be required if all traffic were centrally coordinated. Gaitonde and Tardos (EC 2020 and JACM 2023) initiated the study of such systems, modeling them as an infinitely repeated game in which routers compete for servers and the system maintains a state (the number of packets held at each queue) that results from outcomes of previous rounds. However, their model assumes that servers have no buffers at all, so routers have to resend all packets that were not served successfully, which makes their system model unrealistic. They show that in their model, even with hugely increased server capacity relative to what is needed in the centrally coordinated case, ensuring that the system is stable requires the use of timestamps and priority for older packets.
\vspace{2pt}

We consider a system with two important changes, which make the model more realistic and allow for much higher traffic rates: first, we add a very small buffer to each server, allowing the server to hold on to a single packet to be served later (if it fails to serve it immediately), and second, we do not require timestamps or priority to older packets. Using theoretical analysis and simulations, we show that when queues are learning, a small constant-factor increase in server capacity, compared to what would be needed if centrally coordinating, suffices to keep the system stable, even if servers select randomly among packets arriving simultaneously.

\end{abstract}

\section{Introduction}\label{sec:intro}
In this paper, we model routers in networking as agents using simple learning algorithms to find the best way to get their packets served. This simple, myopic, and distributed multi-agent decision system is an important application of multi-agent learning. Managing a networking system in such a distributed way makes large queuing systems simple to operate, but at the same time, the system needs more capacity than would be required if all traffic were centrally coordinated. In a recent paper, \cite{DBLP:conf/sigecom/GaitondeT20,DBLP:journals/jacm/GaitondeT23}  initiated the study of such systems, modeling them as an infinitely repeated game in which routers compete for servers, and the system maintains a state (number of packets held by each queue) that results from outcomes of previous rounds. Routers get to send a packet at each step in which they received a packet or have one in their queue,\footnote{Players are routers with states (their queue length); the terms {\em router} and {\em queue} are used interchangeably.} to one of the servers. Each server attempts to process only one of the packets that arrive to it. However, the model of \cite{DBLP:conf/sigecom/GaitondeT20,DBLP:journals/jacm/GaitondeT23} (see also \cite{DBLP:conf/nips/SentenacBP21-short,DBLP:conf/nips/0001LW23-short}) assumes that servers have no buffers at all, so queues have to resend all packets that were not served successfully.  They show that, in their system, even with hugely increased server capacity relative to what is needed in the centrally-coordinated case, 
the use of credible timestamps and priority for older packets by all servers is required to ensure that the system is stable, 
i.e., that queue lengths do not diverge with time (see Section~\ref{sec:model} for a formal definition of stability). 

In networking systems, servers accepting packets typically have a very small buffer and can hold on to a few packets to be served later. In this paper, we consider the analog of the model of  \cite{DBLP:conf/sigecom/GaitondeT20,DBLP:journals/jacm/GaitondeT23} in such systems with a tiny buffer at each server. We show that even with a buffer capable of holding only a single packet and non-coordinated learning agents, the service capacity of the system greatly increases.
However, the tasks of the learning algorithms become more complex when servers have buffers. To see this, consider the case of a single queue with many possible servers. Without buffers, the queue faces a classical multi-arm bandit problem, aiming to learn which of the servers is the best one to send their packets to. By contrast, we observe that with a buffer of even just one packet at each server, a sufficiently large number of low-capacity servers will always ensure stability, as the learning agent can now take advantage of these low-capacity servers as well, sending packets to them rarely enough that they are likely to clear in time before the next packet is sent to them. 
This makes the goal of a learning algorithm more complex even with just one queue and many servers: it should learn to distribute its packets in a way that takes advantage of the capacity of all the servers. Interestingly, our results show that even simple bandit learning algorithms manage to learn to use low-capacity servers effectively in the presence of buffers.

Our main result, formally stated in Theorem~\ref{thm:main}, is to show that when multiple agents compete for service while using learning algorithms to identify good servers, a small constant-factor increase in server capacity---relative to what would be required under centralized coordination---suffices to maintain the system stable, even if servers randomly select among simultaneously arriving packets.
Specifically, suppose that we have $n$ queues with packet arrival rates of $\lambda_1,\ldots,\lambda_n$, and $m$ servers with service rates $\mu_1,\ldots,\mu_m$. Obviously, we must have $\sum_i\lambda_i < \sum_j \mu_j$, else the system cannot be stable even if fully coordinated. We show that this condition is not strong enough to ensure stability even with full coordination due to the small size of the buffers.\footnote{We note that the condition allows to create a stable fully coordinated schedule, assuming we use large buffers at the servers. To see how to do this, note that the condition allows to create a fractional assignment such that $\sum_j x_{ij} >\lambda_i$ and $\sum_i x_{ij}<\mu_j$. Using the matching decomposition of this assignment in a coordinated schedule will guarantee a bounded expected number of packets both at the queues  and at the servers.} 
However, our result shows that if $\lambda_i<\frac12$ for all $i$, all queues use any type of learning that achieves the no-regret condition with high probability, and $3\sum_i\lambda_i < \sum_j \mu_j$, then this guarantees that the system remains stable. 
While we do not know whether the factor of 3 in the required capacity constraint is tight and this remains as an open question, we provide a lower bound, showing that to guarantee that a no-regret outcome of learning by the queues remains stable, it is required to have at least $2\sum_i\lambda_i < \sum_j \mu_j$.

\paragraph{Related Work.}\label{sec:related-work}
The classical focus of work on scheduling in queuing systems is aimed at finding schedules that achieve optimal throughput (see, e.g., the textbook of \cite{queuing_theory}). For work evaluating efficiency loss due to selfishness in different classical queuing systems, see the book of \cite{hassin2020rational} and survey of \cite{hassin2003queue}. Closer to our motivation, there is a growing literature that aims to understand how systems perform when the queues use simple learning algorithm to find good service.  \cite{DBLP:conf/nips/KrishnasamySJS16} considers a queue using a no-regret learning algorithm to find what may be the best servers, but does not consider multiple queues competing for service. Their primary goal is to study the expected size of the queue of packets as a function of time. They also extend the result to the case of multiple queues scheduled by a single coordinated scheduling algorithm, assuming there is a perfect matching between queues and optimal servers that can serve them. \cite{DBLP:conf/nips/SentenacBP21-short} and \cite{DBLP:conf/colt/FreundLW22} extended this work to decentralized learning dynamics in bipartite queuing systems that attain near-optimal performance without the matching assumption. For a survey on the role of learning and information in queuing systems see \cite{doi:10.1287/educ.2021.0235}. 

In our work, we diverge from the literature on queuing and scheduling and instead focus on a game-theoretic model with learning agents: 
queues using learning algorithms to best distribute their packets to get good service, while also repeatedly interacting with other learners and competing for servers. We assume that each queue separately learns to selfishly make sure its own packets are served at the highest possible rate, offering a strategic model of scheduling packets in a queuing system. Closest to our model from this literature is the work \cite{DBLP:conf/sigecom/GaitondeT20,DBLP:journals/jacm/GaitondeT23}, who consider the same bipartite model of queues and servers as we do, but without buffers. They show the exact condition to make such a system stable with central coordination of packets and prove that no-regret learning by the queues guarantees stability of the system if it has double the capacity needed for central coordination, assuming packets carry a timestamp and servers choose the oldest packet to serve.  
 \cite{DBLP:conf/wine/FuHL22} extended this work to a general network.
 \cite{DBLP:journals/corr/abs-2302-03614}  propose an alternative, episodic queuing system where agents have incentives to hold jobs in an episode before sending to a central server, but suffer penalties should their jobs not be completed before the end of the episode. They show that both equilibrium and no-regret outcomes ensure stability, so long as these costs are sufficiently large. 

In the works discussed so far, the servers receiving the packets have no buffers: all unserved packets are held in a queue and get resent to be served later. 
In queuing systems aimed at modeling networking, the receivers (servers) each do have a very small buffer, and can hold on to a few packets to be served later. It turns out that even having a buffer for a single packet significantly changes the effective service capacity of the network with learning agents, as already explained above. 

In the case without such buffers, it was feasible to exactly characterize the capacity needed to be able to make the system stable with central coordination by an elegant use of linear programming duality (see \cite{DBLP:conf/wine/FuHL22}). 
In the presence of limited buffer capacity at the servers, the rate at which a server accepts a packet depends on the state of the buffers. This is in contrast to a system without buffers, where the probability that a server accepts a packet is equal to the service rate of the server, and is stable over time. There is some recent work considering queuing systems with such evolving service probabilities, see for example \cite{DBLP:journals/sigmetrics/GrosofHHS24,DBLP:journals/corr/abs-2405-04102}. The goal of that work is to characterize the exact condition that allows the system to remain stable with central coordination of the schedule. Our work is the first to consider stability in such a system assuming each queue independently aims to optimize its own service rate running a no-regret learning algorithm. 

There is a large literature on bounding the price of anarchy for various games, beginning with \cite{DBLP:conf/stacs/KoutsoupiasP99}. These results extend to analyzing outcomes in repeated games assuming players are no-regret learners \cite{DBLP:journals/siamcomp/BalcanBM13,DBLP:journals/jacm/Roughgarden15,DBLP:conf/stoc/SyrgkanisT13}. However, these works make the strong assumption that games in different rounds are independent. By contrast, different rounds in queuing games are not independent: the number of packets in the system depends on the success in previous periods. Work on games with carryover effects, or context, are considered in multi-agent reinforcement learning as well as in Markov games, see for example \cite{littman1994markov,busoniu2008comprehensive,zhang2021multi}. However, this line of work does not focus on the overall quality of learning outcomes. Another line of work on repeated games with such carryover effects between rounds, focusing on such overall learning outcomes, is the repeated ad-auction game with limited budgets, where the remaining budgets make the games no longer independent. See for example \cite{DBLP:journals/mansci/BalseiroG19,DBLP:conf/innovations/GaitondeLLLS23,fikioris2023liquid,fikioris2024learning}.

\section{Model}\label{sec:model}
We consider a more realistic version of the model used by \cite{DBLP:conf/sigecom/GaitondeT20,DBLP:journals/jacm/GaitondeT23}. There is a system of $n$ queues and $m$ servers. We divide time into discrete time steps $t=0,1, \dots$, where in every step $t$, the following occurs. See also the illustration in Figure \ref{fig:model}.
\begin{enumerate}[leftmargin=*, itemindent=0em]
    \item Packets arrive at each queue $i$ according to a fixed probability $\lambda_i$. Formally, we model this by defining $B^i_t \sim \text{Bern}(\lambda_i)$ as an independent random variable that indicates whether a packet arrives at queue $i$ at time $t$. After arrival, each queue that contains at least one packet selects a server to which it attempts to send a single packet.

    \item Each server $j$ processes incoming packets as follows:
    (i) If the server's buffer is full, all received packets are rejected and returned to their respective queues. (ii) Otherwise, the server accepts one of the received packets, chosen uniformly at random, and places it in its buffer, rejecting the remaining packets. 
    (iii) The server then attempts to process the packet in its buffer. With probability $\mu_j$, the processing succeeds, and the packet is removed from the buffer. If the processing fails, the packet remains in the buffer.
        
    \item Any packets not accepted into a server's buffer are returned to their respective queues. Each queue receives bandit feedback based on the outcome of sending its packet: a reward of $0$ if the packet was not accepted by the server or $1$ if the packet was successfully placed in the server's buffer.\footnote{Feedback 
    is immediate after a packet is placed in the buffer, not necessarily when it is successfully processed.}
    
\end{enumerate}
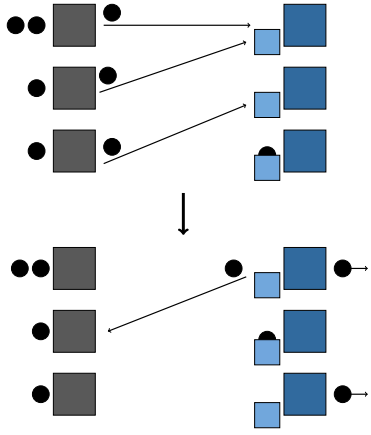
\begin{figure}[t!]
\vspace{-18pt}
\centering
\begin{minipage}{0.3\textwidth}
\resizebox{\textwidth}{!}{%
\begin{tikzpicture}
\tikzstyle{every node}=[font=\LARGE]
\draw [ fill={rgb,255:red,49; green,106; blue,158} ] (12.5,9.75) rectangle  node {\LARGE 
} (15,7.25);
\draw [ fill={rgb,255:red,49; green,106; blue,158} ] (12.5,6) rectangle (15,3.5);
\draw [ fill={rgb,255:red,49; green,106; blue,158} ] (12.5,2.25) rectangle (15,-0.25);
\draw [ fill={rgb,255:red,89; green,89; blue,89} ] (-1.25,9.75) rectangle (1.25,7.25);
\draw [ fill={rgb,255:red,89; green,89; blue,89} ] (-1.25,6) rectangle (1.25,3.5);
\draw [ fill={rgb,255:red,89; green,89; blue,89} ] (-1.25,2.25) rectangle (1.25,-0.25);
\draw [line width=2pt, ->] (1.75,8.5) .. controls (6,8.5) and (6.25,8.5) .. (10.5,8.5) ;
\draw [line width=2pt, ->] (1.5,4.5) -- (10.25,7.5);
\draw [line width=2pt, ->] (1.75,0.25) -- (10.25,3.75);
\draw [ fill={rgb,255:red,0; green,0; blue,0} ] (2.25,9.25) circle (0.5cm);
\draw [ fill={rgb,255:red,0; green,0; blue,0} ] (2,5.5) circle (0.5cm);
\draw [ fill={rgb,255:red,0; green,0; blue,0} ] (2.25,1.25) circle (0.5cm);
\draw [ fill={rgb,255:red,0; green,0; blue,0} ] (-2.25,8.5) circle (0.5cm);
\draw [ fill={rgb,255:red,0; green,0; blue,0} ] (-2.25,4.75) circle (0.5cm);
\draw [ fill={rgb,255:red,0; green,0; blue,0} ] (-2.25,1) circle (0.5cm);
\draw [ fill={rgb,255:red,0; green,0; blue,0} ] (-3.5,8.5) circle (0.5cm);
\draw [line width=5pt, ->] (6.5,-1.5) -- (6.5,-4);
\node [font=\LARGE] at (0.5,8) {};
\node [font=\LARGE] at (0.5,8) {};
\draw [ fill={rgb,255:red,89; green,89; blue,89} ] (-1.25,-4.75) rectangle (1.25,-7.25);
\draw [ fill={rgb,255:red,89; green,89; blue,89} ] (-1.25,-8.5) rectangle (1.25,-11);
\draw [ fill={rgb,255:red,89; green,89; blue,89} ] (-1.25,-12.25) rectangle (1.25,-14.75);
\draw [ fill={rgb,255:red,49; green,106; blue,158} ] (12.5,-4.75) rectangle (15,-7.25);
\draw [ fill={rgb,255:red,49; green,106; blue,158} ] (12.5,-8.5) rectangle (15,-11);
\draw [ fill={rgb,255:red,49; green,106; blue,158} ] (12.5,-12.25) rectangle (15,-14.75);
\draw [ fill={rgb,255:red,112; green,168; blue,219} ] (10.75,8.25) rectangle (12.25,6.75);
\draw [ fill={rgb,255:red,112; green,168; blue,219} ] (12.25,3) rectangle (10.75,4.5);
\draw [line width=2pt, ->] (10.25,-6.5) -- (2,-9.75);
\draw [ fill={rgb,255:red,0; green,0; blue,0} , line width=2pt ] (11.5,0.75) circle (0.5cm);
\draw [ fill={rgb,255:red,113; green,167; blue,220} ] (10.75,0.75) rectangle (12.25,-0.75);
\draw [ fill={rgb,255:red,113; green,167; blue,220} , line width=0.5pt ] (10.75,-6.25) rectangle (12.25,-7.75);
\draw [ fill={rgb,255:red,113; green,167; blue,220} , line width=0.5pt ] (10.75,-14) rectangle (12.25,-15.5);
\draw [ fill={rgb,255:red,0; green,0; blue,0} , line width=0.5pt ] (11.5,-10.25) circle (0.5cm);
\draw [ fill={rgb,255:red,0; green,0; blue,0} , line width=0.5pt ] (16,-13.5) circle (0.5cm);
\draw [ fill={rgb,255:red,0; green,0; blue,0} , line width=0.5pt ] (16,-6) circle (0.5cm);
\draw [ fill={rgb,255:red,0; green,0; blue,0} , line width=0.5pt ] (9.5,-6) circle (0.5cm);
\draw [ fill={rgb,255:red,0; green,0; blue,0} , line width=0.5pt ] (-2,-6) circle (0.5cm);
\draw [ fill={rgb,255:red,0; green,0; blue,0} , line width=0.5pt ] (-2,-9.75) circle (0.5cm);
\draw [ fill={rgb,255:red,0; green,0; blue,0} , line width=0.5pt ] (-2,-13.5) circle (0.5cm);
\draw [ fill={rgb,255:red,0; green,0; blue,0} , line width=0.5pt ] (-3.25,-6) circle (0.5cm);
\draw [line width=2pt, ->] (16.25,-6) -- (17.5,-6);
\draw [line width=2pt, ->] (16,-13.5) -- (17.5,-13.5);
\draw [ fill={rgb,255:red,113; green,166; blue,221} ] (10.75,-10.25) rectangle (12.25,-11.75);
\end{tikzpicture}
}
\end{minipage}%
\hspace{0.5cm}%
\begin{minipage}{0.6\textwidth}
\captionof{figure}{A system with three queues and three servers in a single time step. Discs represent packets; the left squares are queues, the right squares are servers, and the small squares are buffers. In this time step, two packets are sent to the top server. One of these packets is accepted into the buffer and then processed, while the other is returned to its queue. For the other two servers, one admits a packet into its buffer but fails to serve it, and the other processes a packet that was already in its buffer from a previous round.}
\label{fig:model}
\end{minipage}
\vspace{-1pt}
\end{figure}

\begin{definition}
    We say that a system with a schedule for sending packets is \emph{stable} if the expected number of packets waiting to be sent or served in the system is bounded by a constant at all times. We allow the constant to depend on the number of queues or servers and the parameters of the system, but it may not grow with time $T$. 
\end{definition}

It is clear that if $\sum_i \lambda_i >\sum_j \mu_j$ then even a fully coordinated system cannot be stable, and the total number of packets waiting at the queues must grow linearly in time. In the case of a single queue with a single server, $\lambda <\mu$ is the necessary and sufficient condition for keeping the packets remaining in the system bounded by a constant (which depends on the gap in the inequality) at all times; see, e.g.,~\cite{queuing_theory}. 
The analogous condition  in the case of multiple servers is clearly necessary but not sufficient. 
\begin{lemma}
    Consider an example with one queue and two servers, each with service probability $\mu=\frac12$. We claim that $\lambda < \frac{23}{24}<1$ is needed for the system to be stable even with full coordination.
\end{lemma}
\begin{proof}
    To see why this is true, consider 3 consecutive time steps. To analyze the maximum processing rate, we look at the scenario where the queue sends a packet in each of these steps. We claim that with probability at least $1/8$ it holds that in one of the three times, the queue must fail to have a packet accepted by a server. 
    Assume that the packet sent in the first of the three periods is accepted into a buffer but not processed for two steps. This event has probability at least $\frac14$. With coordination, a packet in the next period will be sent to the other server. Assume that this packet is not processed when it is sent, which has probability at least $\frac12$. Now in the third period, both buffers are full, so no packet can be accepted. 
    
This proves that during three consecutive steps, the maximum expected service rate is at most $\frac{\frac18\cdot 2+\frac78 \cdot 3}{3}=\frac{23}{24}.$
Thus, if the arrival rate $\lambda$ is higher, we have a buildup of the queue linear in $T$.  
\end{proof}

\section{Analysis}\label{sec:analysis}
Our main result is the following theorem. This result holds under either bandit or full feedback, and applies to arbitrary learning algorithms, as long as they achieve sublinear regret.

\begin{theorem}
\label{thm:main}
Assuming $\sum_i\lambda_i < \frac{1}{3} \sum_j \mu_j$, $\lambda_i < \frac12$ for all $i$, and all queues use a form of learning guaranteeing 
low regret with high probability to identify servers they can use, the system remains stable with the expected number of packets in the system bounded by a (time-independent) constant at all times. 
\end{theorem}

The key tool in our analysis is the following theorem of \cite{pemantle1999moment}. Informally, the theorem states that a stochastic sequence that has negative drift when it grows large and is sufficiently regular (has bounded moments) is bounded by a constant at all times.

\begin{theorem}
\label{thm:pemantle}
Let $X_1,X_2,\ldots$ be a sequence of nonnegative random variables with the property that
\begin{enumerate}[leftmargin=*, itemindent=0em]
    \item There exist constants $\alpha,\beta>0$ such that  
    $
        \mathbb{E}[X_{t+1}-X_t\vert \mathcal{F}_t \& X_t\ge \beta]<-\alpha,
    $ 
    where $\mathcal{F}_t$
    denotes the history
    until period $t$. 
    
    \item There exist $p>2$ and a constant $\theta>0$ such that for any history, 
    $ 
        \mathbb{E}[\vert X_{t+1}-X_t\vert^p\vert \mathcal{F}_t]\leq \theta.
    $
\end{enumerate}
Then, for any $0<r<p-1$, there exists an absolute constant $M=M(\alpha,\beta,\theta,p,r)$ not depending on $t$ such that $\mathbb{E}[X_t^r]\leq M$ for all $t$.
\end{theorem}

To use this theorem, we consider a long enough time period of length $T$ and use a potential function that depends on the number of packets in each of the queues at the start of that time period. Assume that queue $i$ has $N_i$ packets at the start of the period. We will use the potential function $\Phi=\sum_i (N_i-(\frac{1}{2}\lambda_i+2\delta) T)^+$, where $x^+=\max(x,0)$, and will aim to show that between the beginning and end of a $T$ long period, the change in this potential satisfies the conditions of Theorem \ref{thm:main}.

For $T$ large enough, we expect the empirical packet arrivals and servers' service rates are all close enough to their expectations, and each queue has small enough regret for their learning algorithm. We will call a period of length $T$ \emph{good}$_\delta$ for a parameter $\delta>0$ (to be chosen later) if the following three conditions are all satisfied. Later, we will also account for what happens in the ``bad event'' when \emph{good}$_\delta$ does not occur.

\begin{condition}\label{con:queue_arrival}
    The number of packets arriving at each queue $i$ during each period of length $\hat T \le T$ is at most $\lambda_i\hat T+\delta T$.
\end{condition}

\begin{condition}\label{con:service_rate}
    Each server $j$ that has a packet in its buffer (possibly one that just arrived) for at least half of the steps during this $T$-length period succeeds in serving at least $(\frac12\mu_j-\delta) T$ packets. 
\end{condition}

\begin{condition}\label{con:learning}
    The learning algorithm of each queue $i$ accumulates regret at most $\delta T$ during the $T$-length period for the decisions it made during this period when it had a packet to send.
\end{condition}

To show the expected decrease in the potential function, we consider whether a period of length $T$ satisfied the condition of being 
\emph{good}$_\delta$. 
We will show below (in Lemmas~\ref{lm:Chernoff} and~\ref{lm:learning}) that the probability that \emph{good$_\delta$} fails to hold is bounded by $(m+2n)\eta$ for a small $\eta$, as shown by the two lemmas.
The maximum possible increase in the potential function during a period of $T$ steps is at most $nT$ (if packets arrive during each step at each of the queues, and none reach the servers), so bad events contribute at most $(m+2n)\eta nT$ to the expected change in potential over the $T$ steps. 

To show the decrease in the potential in a \emph{good$_\delta$} period of length $T$, the following two lemmas separately consider servers that serve packets most of the time and those that do not.

\begin{lemma}
\label{lm:full_server}
  If all servers have a packet in their buffer (possibly one that just arrived) for more than half of the iterations, and we are in the \emph{good$_\delta$} case, then the total number of packets at the queues decreases by at least  $(\frac12 \sum_j \mu_j-\sum_i\lambda_i -(n+m)\delta) T$.
\end{lemma}

\begin{proof}
By Condition \ref{con:service_rate}, the total number of packets served is at least $\sum_j (\frac12\mu_j -\delta)T$. By Condition \ref{con:queue_arrival}, the total number of arriving packets is at most $\sum_i (\lambda_i+\delta)T$. Combining these two bounds establishes the lemma. 
\end{proof}

\begin{lemma} \label{lm:open_server}
Suppose there is a server whose buffer is empty more than half of the $T$ iterations, and we are in the $\textit{good}_\delta$ case. Consider a queue $i$ and assume that $\delta < \frac12 (\frac12 - \lambda_i)$.
Then, either the number of packets at this queue decreases during the period of length $T$, or the number of packets left in it at the end of the period is at most $N'_i \leq (\frac{1}{2} \lambda_i+2\delta) T$.
If at the start of the period there are at least $T$ packets in the queue, then the number decreases by at least $(\frac{1}{2} - \lambda_i - 2\delta) T$.
\end{lemma}
\begin{proof}
If the queue has packets in every step in the window, by the no-regret Condition \ref{con:learning}, it clears at least $(\frac{1}{2}-\delta)T$.
By Condition \ref{con:queue_arrival}, it receives at most $(\lambda_i+\delta) T$ packets, so the number of packets goes down by at least $(\frac{1}{2} - \lambda_i-2\delta) T$.
This happens, in particular, when the queue has at least $T$ packets at the beginning of the window.

Now suppose the queue has packets in the last $\hat{T}$ steps, but was empty in the previous step.
If $\hat{T} \leq \frac{1}{2}T$, then by Condition \ref{con:queue_arrival}, in the last $\hat{T}$ steps it receives at most $(\frac{1}{2} \lambda_i + \delta) T$ new packets.
So $N'_i \leq (\frac{1}{2} \lambda_i + \delta) T$.

The rest of the proof focuses on the case when $\hat{T} > \frac{1}{2}T$.
Let $I$ be the number of steps queue $i$ and the server were empty. So the server was empty at least $\frac12 T-I$ steps that queue $i$ was attempting to send packets. By the no-regret Condition \ref{con:learning}, queue $i$ must have succeeded in clearing close to this many packets.
Let $S$ be the number of packets queue $i$ received before the last $\hat{T}$ steps.
Before the last $\hat{T}$ steps, queue $i$ had and cleared $N_i+S$ packets, so it was empty at most $T - \hat{T} - N_i - S$ steps, and this is also an upper bound for $I$. 
Using this upper bound with the no-regret Condition \ref{con:learning}, we get that queue $i$ cleared at least $\frac12 T - I -\delta T \geq  \hat{T} - (\frac12+\delta) T + N_i + S$ packets.
Queue $i$ had in total $N_i+S$ packets before the last $\hat{T}$ steps, where by Condition \ref{con:queue_arrival} it received at most $\lambda_i \hat{T} + \delta T$ packets.
Therefore:
\begin{eqnarray*}
    N_i' &\leq& \left(N_i+S+\lambda_i \hat T + \delta T\right)-\left(\hat{T} - \left(\frac12+\delta\right) T + N_i + S\right)\\
    & = &\left(\frac12 + 2 \delta\right) T + (\lambda_i-1) \hat{T} \leq \left(\frac12 + 2 \delta\right) T + \frac12(\lambda_i-1) T =
    \left(\frac12 \lambda_i + 2 \delta \right) T,
\end{eqnarray*}
where the last inequality follows from $\lambda_i < 1$ and $\hat{T} > \frac12 T$.
\end{proof}

To prove Theorem~\ref{thm:main} we now need to bound the probability that \emph{good$_\delta$} fails to hold. 
We use Chernoff and union bounds to bound from below the probability that conditions~\ref{con:queue_arrival} and~\ref{con:service_rate} hold when choosing a sufficiently large $T$.

\begin{lemma}
\label{lm:Chernoff}
For any constants $\delta>0$ and $\eta>0$, by choosing the period length $T$ high enough, we can guarantee that the probability that conditions \ref{con:queue_arrival} or \ref{con:service_rate} fail to hold for a single queue or single server is bounded by $\eta$, and hence with probability $(1-\eta(n+m))$ the conditions hold. 
\end{lemma}

To guarantee that Condition \ref{con:learning} holds with high probability, we will assume that queues use learning algorithms that have high-probability regret guarantees: 
Let $R_T$ be the regret of a queue. We say that a queue has low regret with high probability if for every  $\eta > 0$, there exists a sublinear bound $\epsilon(T) = o(T)$ such that $Pr[R_T \leq \epsilon(T)] > 1 - \eta$.
This can be guaranteed, for example, using the algorithm EXP3.P.1 of \cite{DBLP:journals/siamcomp/AuerCFS02}.

\begin{lemma}
\label{lm:learning}
Suppose that all the queues use learning methods that guarantee low regret with high probability.
Then, for every $\delta>0$ and $\eta>0$, by choosing the period length $T$ high enough we can guarantee that the probability that Condition \ref{con:learning} holds for a single queue is at least $1-\eta$. Hence, the condition holds with probability at least $(1-n\eta).$
\end{lemma}

We are now ready to start proving the main result.

\begin{proof} (Theorem \ref{thm:main}) 
Recall the potential function $\Phi=\sum_i (N_i-(\frac12 \lambda_i+2\delta) T)^+$. Clearly, $\Phi$ has bounded moments, as the maximum possible increase in the potential function over a period of length $T$ is when a packet arrives at every single queue during each step and no packets are accepted at any of the servers. That is a total increase of $nT$. Similarly, the maximum possible decrease in the potential function over a period of length $T$ is when no packets arrive, and each queue succeeds in sending a packet to a server at each step. That is a total decrease of $nT$. 

Next, we want to prove that when $\Phi>Tn$, the potential is expected to decrease. Recall that when \emph{good}$_\delta$ fails to hold the potential can increase by as much as $nT$. By Lemmas \ref{lm:Chernoff} and \ref{lm:learning}, this can contribute at most $\eta (m+2n)nT$ to the expectation.

To analyze the case when \emph{good}$_\delta$ holds, we consider two cases separately. If all servers have packets to try to serve (in their buffer, or just arriving) at least half of the time, then by Lemma \ref{lm:full_server} the total number of packets decreases by at least $(\frac12 \sum_j \mu_j -\sum_i\lambda_i -(n+m)\delta)T$. Some of this decrease may not affect the potential function, since some of the cleared packets can come from queues that contribute zero to the potential. However, this can only affect $\sum_i (\frac12\lambda_i+2\delta) T$ packets, and at most $(\frac12\lambda_i+2\delta) T$ packets in queue $i$, and so the potential is decreasing by at least 
\begin{equation*}
\Big(\frac12 \sum_j \mu_j -\frac32 \sum_i \lambda_i -(3n+m)\delta \Big)T.
\end{equation*}
By the assumption that $\sum_j \mu_j >3\sum_i \lambda_i$, this is an $\Omega(T)$ decrease, if we choose a small enough $\delta$.

Now consider the case where a server is open (i.e., its buffer is empty, and no new packet arrived) half the time. If $\Phi>nT$ then some queue $i$ will have at least $T$ packets. In this case, by Lemma \ref{lm:open_server}, this queue will have its packet count decrease by at least
$(\frac12 -\lambda_i-2\delta)T$, contributing this decrease to the potential function. Lemma \ref{lm:open_server} also shows that all other queues will have their contribution to the potential decrease, unless it is zero at the end of the $T$-length period. This shows a total decrease in the potential of at least 
$(\frac12 -\lambda_i-2\delta)T=\Omega(T)$, again assuming $\delta$ is small enough. 

Putting together the expected increase when \emph{good}$_\delta$ fails to hold and expected decrease in the two cases when \emph{good}$_\delta$ holds (and using that $\lambda_i < \frac12$), we get that the potential function is expected to decrease after a $T$-length period by at least
\begin{equation*} \textstyle
T\min(\frac12 -\lambda_i-2\delta, \ \frac12 \sum_j \mu_j -\frac{3}{2} \sum_i\lambda_i  -(3n+m)\delta)
 - T\eta (m+2n)n.
\end{equation*}
To guarantee the required expected decrease of the potential, we choose $\delta$ small enough so that both terms in the $\min$ are at least a positive constant (so in the case  \emph{good}$_\delta$ happens the potential decreases by $\Omega(T)$) and then choose $\eta$ small enough to make the total positive.
\end{proof}

While we do not know if the factor of 3 in the main theorem is a tight worst-case bound, 
the following example shows that we need at least a factor of 2 to ensure stability. 

\begin{theorem}[Lower Bound] \label{thm:lower-bound-example}
Consider a queue with $1/2$ arrival rate, facing $k$ servers, each with only capacity $1/n$. Assume that the server chooses a uniform random server at each step. This sending plan will clearly satisfy the no-regret Condition \ref{con:learning}. We claim that we need at least $n$ servers to make the solution stable proving a lower bound of 2 needed for the needed increase in server capacity. 
\end{theorem}
\begin{proof}
So with $k$ servers, the chance that when the queue sends to a server, the previous time it did so was $i$ steps ago is $1/k \cdot (1-1/k)^{i-1}$. If the previous time the queue did sent there, the chance that the server is still busy is $(1-1/n)^i$, so the chance that the new packet gets accepted is
$$
\sum_{i=1}^{\infty}    1/k \cdot (1-1/k)^{i-1} (1-(1-1/n)^i)
=\frac{1}{k}\sum_{i=1} (1-\frac{1}{k})^{i-1}-\frac{n-1}{nk}\sum_{i=1}^{\infty} \big((1-\frac{1}{k})(1-\frac{1}{n})\big)^{i-1}
$$
$$
=1-\frac{n-1}{nk}\frac{1}{1-(1-1/k)(1-1/n)}
=1-\frac{n-1}{n+k-1}.
$$
To make this stable, we need this value to be above $1/2$, so we need 
$k>n-1$. This results in a total capacity of $\frac{k}{n}>\frac{n-1}{n}$, essentially double the arrival rate of the queue.
\end{proof}

\section{Experiments}\label{sec:experiments}
\begin{figure*}[ht]
\vspace{-20pt}
    \centering
    \begin{subfigure}{.42\textwidth}
        \includegraphics[width=0.9\linewidth]{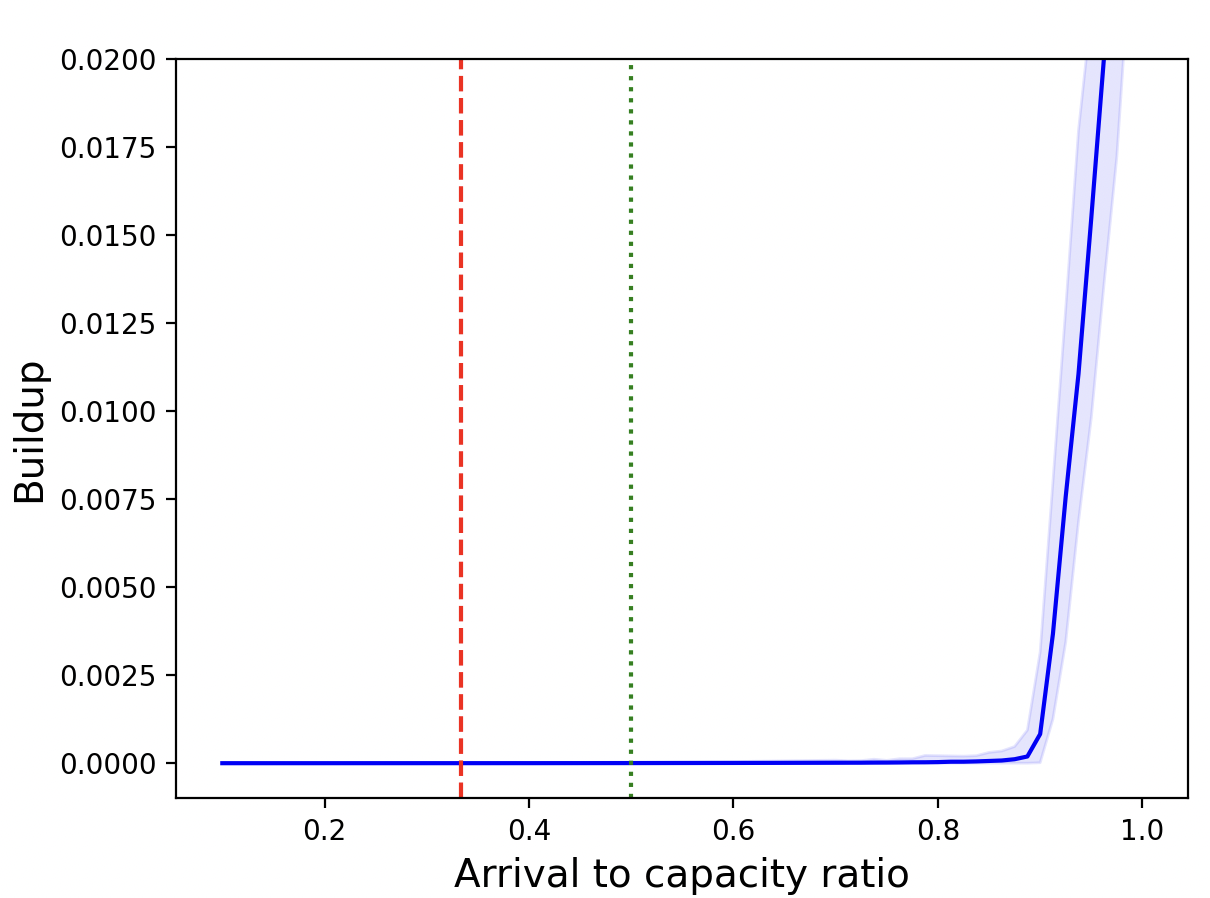} 
        \caption{Buildup for a symmetric system with 3 servers with capacity $\mu$ and 3 queues with arrival rates $\lambda$ as a function of $\lambda/\mu$.}
        \label{fig:fig1a-buildup-vs-capacity-ratio-symmentric}
    \end{subfigure}
    \hfill
    \begin{subfigure}{.42\textwidth}
        \includegraphics[width=0.9\linewidth]{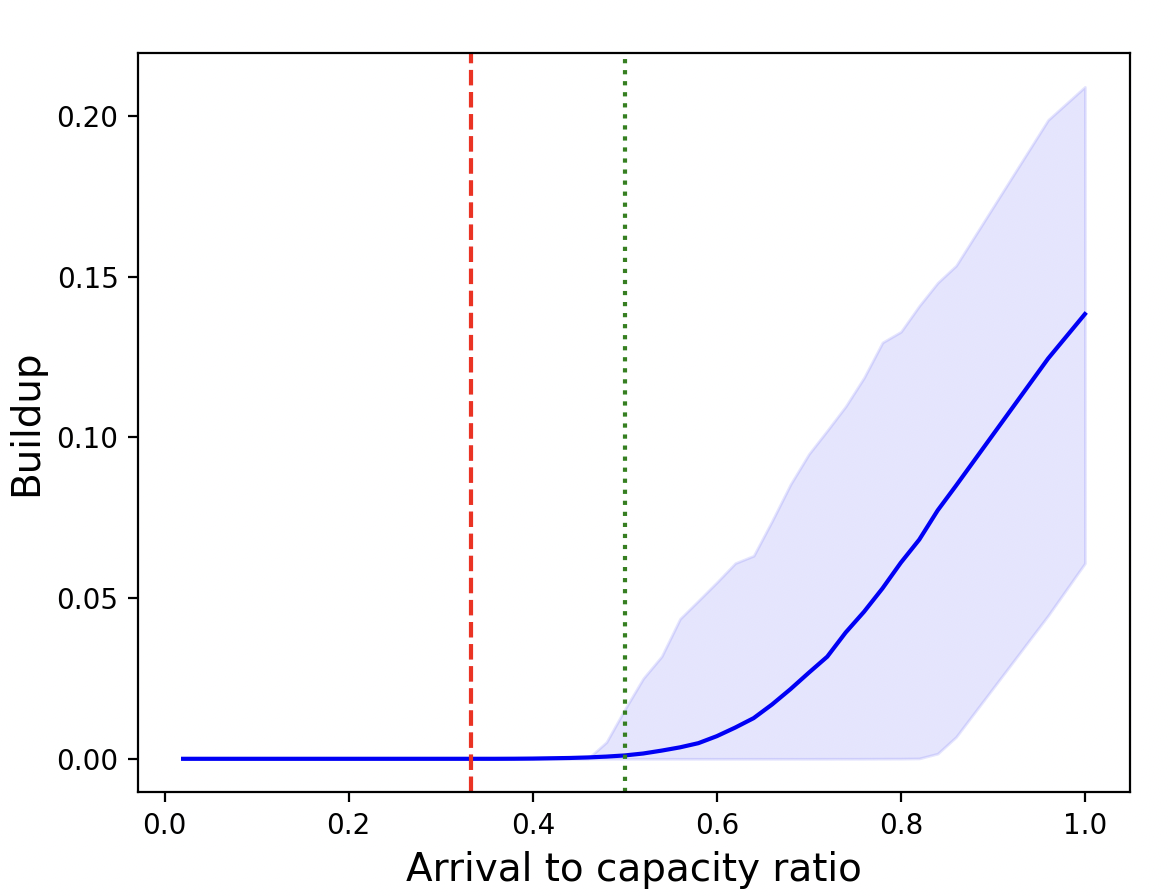} 
        \caption{Buildup for an ensemble of 200 randomized systems with $5$ queues and $6$ servers as a function of $\sum_i \lambda_i/\sum_j \mu_j$.}
        \label{fig:fig1b-buildup-vs-capacity-ratio-non-symmentric}
    \end{subfigure}
    \caption{Empirical buildup of total queue sizes normalized by $n \cdot T$ as a function of the ratio $\sum_i \lambda_i/\sum_j \mu_j$ across different scenarios. The dashed vertical lines mark the ratios of $1/3$ and $1/2$ from our analysis.}
    \label{fig:fig1-buildup-vs-capacity-ratio}
    \vspace{-5pt}
\end{figure*}

In the preceding section, we derived theoretical stability conditions and worst-case bounds on the buildup of queue lengths. Specifically, our proof showed that in a system with singleton buffers, whenever a queue grows large, given that a condition of the total service capacity is satisfied, a no-regret property of the agents selecting servers guarantees that it will shrink at a linear rate, keeping the total lengths of queues bounded. In this section, we complement these theoretical guarantees with computational experiments to observe the typical behavior of these systems beyond the worst case, and explore the empirical impact of buffers on the game dynamics.

We conduct simulations using the EXP3 algorithm \cite{auer2002nonstochastic} (see also \cite{slivkins2019introduction}), implemented in Python. For the simulation parameters, we denote $\gamma$ for exploration rate and $T$ for the time horizon of the simulation. 

\subsection{Empirical Queue Buildup}

We start by studying the relationship between the system's arrival rate and its service capacity. The service capacity must strictly exceed the arrival rate to prevent queue buildup. This observation is seen also in typical experimental scenarios.

Figure \ref{fig:fig1-buildup-vs-capacity-ratio} illustrates the total empirical buildup, which is the sum of all packets left in queues after $T$ iterations, normalized by $n \cdot T$ (vertical axis) as a function of the arrival-to-capacity ratio $\sum_i \lambda_i/\sum_j \mu_j$ (horizontal axis). 
Figure \ref{fig:fig1a-buildup-vs-capacity-ratio-symmentric} focuses on a symmetric system with three servers such that $\mu_1 = 0.8$ and $\mu_2 = \mu_3 = 0.2$ and three queues with equal arrival rates $\lambda_i$, which range between $0.02$ and $0.4$. The simulations run for $T = 50$,$000$ steps with $\gamma = 1/\sqrt{T}$. The solid line represents the average buildup over $200$ independent simulations for each value of $\lambda_i$, and the shaded region indicates the range between minimum and maximum buildup values observed in the simulations. A clear transition emerges: when $\sum_i \lambda_i/\sum_j \mu_j > 0.9$, the buildup becomes proportional to $T$. Crucially, this transition occurs at a point above $0.5$ and below $1$ (where buildup becomes inevitable). 

Figure \ref{fig:fig1b-buildup-vs-capacity-ratio-non-symmentric} presents the same analysis for an ensemble of randomized systems with $5$ queues and $6$ servers. To generate such instances, we proceed as follows: generate $200$ base parameters by selecting 5 queue arrival rates, $\lambda_i$, and 6 server capacity rates, $\mu_j$, uniformly at random in $(0,1)$. For each value of the ratio $r = \sum_i \lambda_i/\sum_j \mu_j$ (with $r$ in $(0, 1]$), rescale the arrival rates such that the ratio $\sum_i \lambda_i/\sum_j\mu_j = r$ is satisfied. 
This process is repeated for each of the $200$ sets of base parameters. The simulations are run for $T = 50$,$000$ and 
$\gamma = 1/\sqrt{T}$. In Figure \ref{fig:fig1b-buildup-vs-capacity-ratio-non-symmentric}, the solid line represents the average buildup over the $200$ simulation for each $r$ value; the shaded region indicates the range between the $2.5$ and $97.5$ percentiles of buildup values. The results confirm that, even for asymmetric and randomized systems, the transition point remains above $0.5$ and below $1$. In both Figure \ref{fig:fig1a-buildup-vs-capacity-ratio-symmentric} and Figure \ref{fig:fig1b-buildup-vs-capacity-ratio-non-symmentric} the vertical lines indicate where $r = \frac{1}{3}$ and $r = \frac{1}{2}$. The empirical buildups in these simulations show that while the worst-case bound is $\sum_i\lambda_i < \frac{1}{3} \sum_j \mu_j$, systems with learning agents typically do better than the worst case at clearing packets.

\begin{figure*}[t!]
\vspace{-24pt}
    \centering
    \begin{subfigure}{.42\textwidth}
        \includegraphics[width=0.9\linewidth]{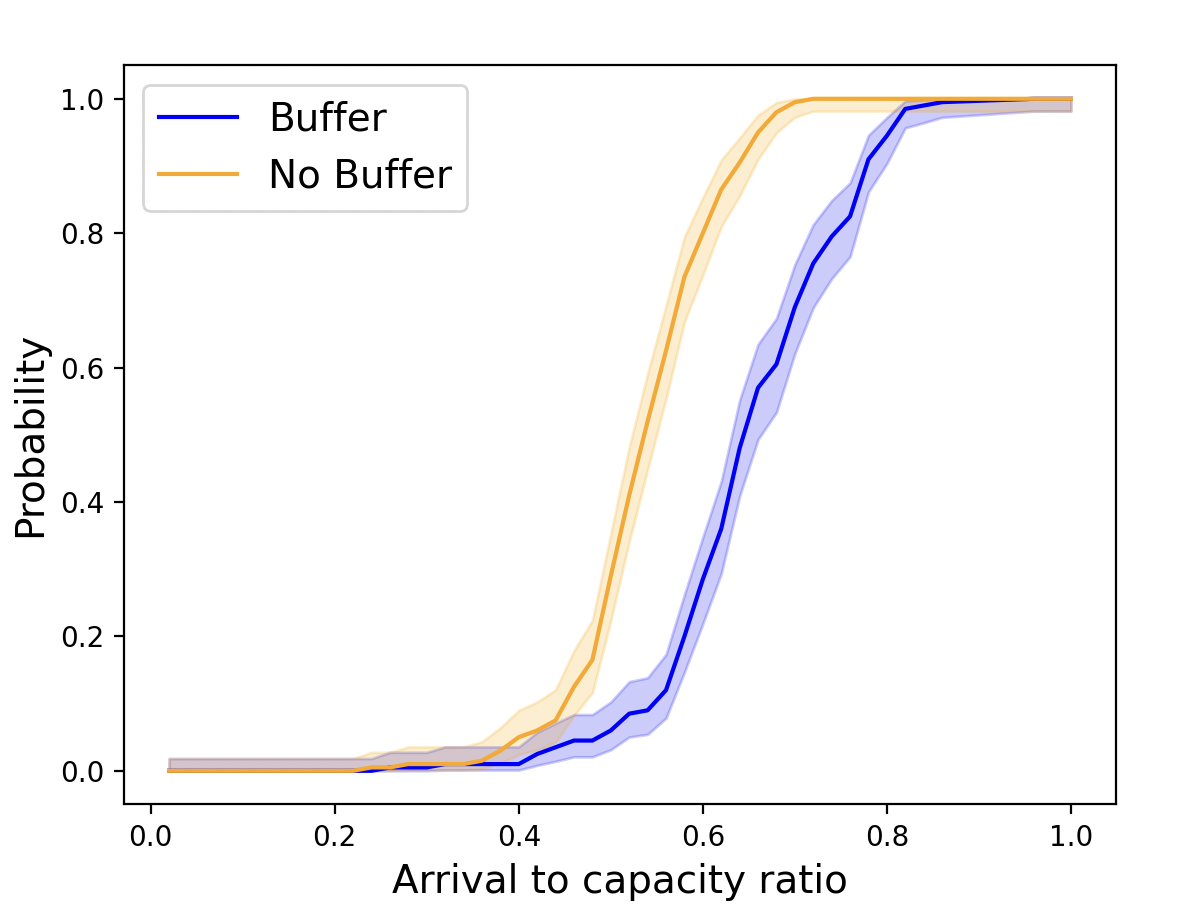} 
        \caption{Probability that a randomized system with time horizon $T$ has a queue with buildup greater than $\sqrt{T}$ as a function of the ratio $\sum_i \lambda_i/\sum_j \mu_j$. Shaded regions show 95\% 
        confidence interval for the probability estimation.}
        \label{fig:fig2a-buildup-with-vs-without-buffers}
    \end{subfigure}
    \hfill
    \begin{subfigure}{.42\textwidth}
        \includegraphics[width=0.9\linewidth]{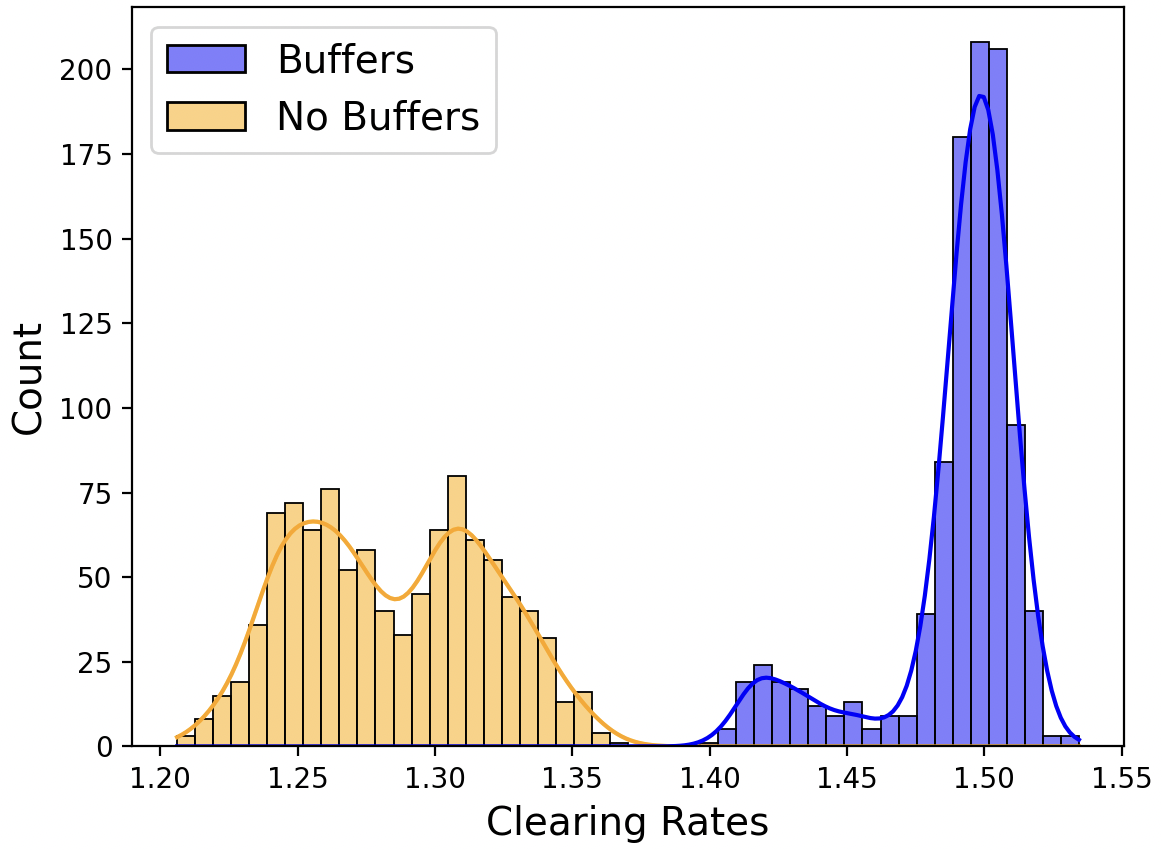} 
        \caption{Empirical clearing rates (the number of packets cleared by a system normalized by time horizon $T$) across $1$,$000$ samples of a system with fixed parameters. The left histogram is without buffers and the right one is with buffers.}
        \label{fig:fig2b-secvice-rate-with-vs-without-buffers}
    \end{subfigure}
    \caption{A comparison of identical systems with and without buffers.}
    \label{fig:fig2-with-vs-without-buffers}
    \vspace{-5pt}
\end{figure*}

\begin{figure*}[t!]
\vspace{-12pt}
    \centering
    \begin{subfigure}{.325\textwidth}
        \includegraphics[width=1.07\linewidth]{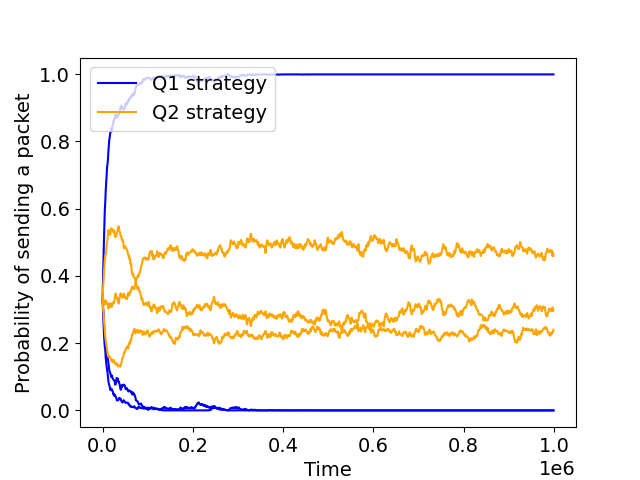} 
        \caption{Dynamics in a system with two queues with arrival rates $\lambda_1 = 1/3$, $\lambda_2 = 1/6$ and three servers with service capacities $\mu_1 = 2/3$, $\mu_2 = 2/9$, $\mu_3 = 1/9$ as a function of time.}
        \label{fig:fig3a-dynamics}
    \end{subfigure}
    \hfill
    \begin{subfigure}{.325\textwidth}
        \includegraphics[width=1.07\linewidth]{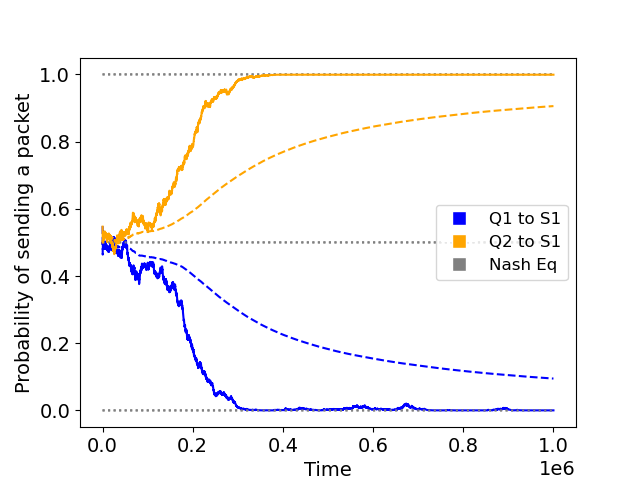} 
        \caption{Probability of sending to server 1 in 
        a system with 
        2 queues with 
        $\lambda_i = 1/4$, and 2 servers with $\mu_i = 2/3$, compared to Nash equilibria. Dashed lines show time averages.}
        \label{fig:fig3b-dynamics-and-NE}
    \end{subfigure}
    \hfill
    \begin{subfigure}{.325\textwidth}
        \includegraphics[width=1.07\linewidth]{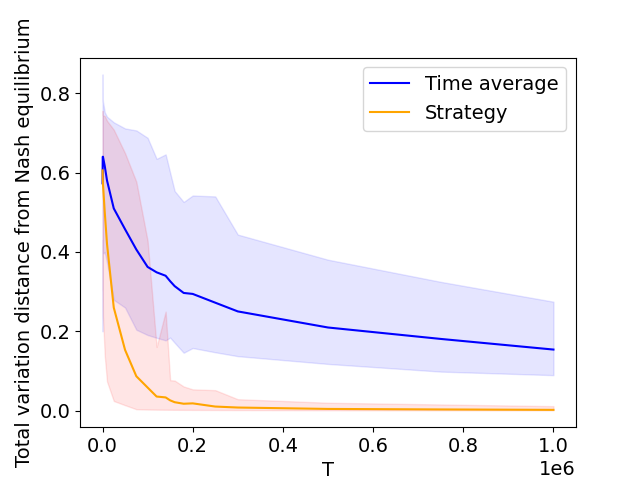} 
        \caption{Total variation distance  
        from a pure Nash equilibrium as a function of time in multiple simulations in the system from Figure \ref{fig:fig3b-dynamics-and-NE}. 
        Shaded regions are 95\% confidence intervals.}
        \label{fig:fig3c-TV-from-NE}
    \end{subfigure}
    \caption{Dynamics of probability distributions and the empirical play across different scenarios.}
    \label{fig:fig3-dynamics}
    \vspace{-5pt}
\end{figure*}

\subsection{The Impact of Buffers}
We now compare systems with vs. without buffers. 
First we note that even a buffer of just one packet at each server can significantly increase server capacity. To see this, consider a system with 
a single queue receiving packets at a high rate, say with probability 
$\lambda=\frac12$ at each time step, and multiple servers, each with a service rate, say
$\mu=\frac13$. Without buffers, the condition for stability of such a system is that there is a single server with a service rate of $\mu>\lambda$, as a queue can only send a packet to one server at a time.  However, with buffers, two such servers will make the system stable if the queue alternates sending its packets to them.

Our empirical methodology for illustrating the impact of buffers follows the approach described in the previous subsection. Figure \ref{fig:fig2a-buildup-with-vs-without-buffers} illustrates the effect of adding buffers to a system in the random-instance scenario, with $n=5$ queues and $m=6$ servers as described earlier. It shows the probability that some queue in the system has at time $T$ a buildup greater than $\sqrt{T}$ as a function of the arrival-to-capacity ratio for systems with and without buffers. The solid lines show the observed frequency in $300$ independent simulations for each value of $r = \sum_i \lambda_i/\sum_j \mu_j$ in $(0, 1]$; the shaded region shows 
the $95\%$ confidence interval in estimating these probabilities. We observe that systems without buffers require a higher capacity relative to the arrival rate to prevent buildup.

Figure \ref{fig:fig2b-secvice-rate-with-vs-without-buffers} shows the 
distributions of packet clearing rates for two identical systems, with and without buffers, over $1$,$000$ simulations with $T=10$,$000$ and $\gamma = 1/\sqrt{T}$. There are $4$ queues with $\lambda_1 = 0.6, \lambda_2 = \lambda_3 = \lambda_4 = 0.3$ and $5$ servers with $\mu_1 = 0.8, \mu_2 = \mu_3 = 0.4, \mu_4 = \mu_5 = 0.2$. The clear separation between the distributions demonstrates that buffers increase efficiency. In the system without buffers, 
queues preferred the servers with higher capacity, which led to more competition and a lower clearing rate since the three high-capacity servers could not process all packets.\footnote{
The total service rate of the three 
top servers
equals
the total arrival rate, hence is not enough for stability.} In the system with buffers, queues do learn to extract good service rates also from the lower capacity servers, which increases the overall clearing rate of the system.

\subsection{Dynamics and Convergence}

Next, we look at the dynamics of the learning agents in our systems. We observe that in all parameter configurations we have tested, the dynamics converge in last iterate (up to a small noise level due to the finite time horizon and step size in our simulations) to Nash equilibria, as illustrated in Figure \ref{fig:fig3-dynamics}.  This is surprising, as even without the carryover effects that our games have, no-regret dynamics need not converge to Nash equilibria.\footnote{In fact, no-regret dynamics may fail to converge to any single coarse correlated equilibrium even in the average-iterate sense \cite{kolumbus2022auctions,kolumbus2022and}.}
Figure \ref{fig:fig3a-dynamics} depicts the dynamics in systems with $2$ queues and $3$ servers. It can be seen that the probabilities of play of the algorithms approximately converge to a stationary distribution. Note that a stationary distribution of play and the no-regret condition imply that this distribution is a Nash equilibrium. 
Figure \ref{fig:fig3b-dynamics-and-NE} shows a dynamic in a simpler system which has two pure Nash equilibria and one mixed equilibrium. As can be seen in the instance in the figure, which is a typical one, the dynamic in this system converges to a pure equilibrium. 
Figure \ref{fig:fig3c-TV-from-NE} depicts the total variation distance from the closest Nash equilibrium as a function of the horizon $T$ in $200$ simulations of the same system for each $T$. The strategies quickly converge, and so the distance of the historical average of play shrinks as well, roughly as $1/\sqrt{T}$.

\section{Conclusion} 
This work contributes to the growing literature on learning in games with carryover effects between rounds. Routing is an important setting where such effects arise and where routers use simple distributed learning algorithms. Carryover effects appear in many other strategic multi-agent settings, such as auctions with budget constraints, or in investment markets where current wealth affects future opportunities. A key observation from our results---potentially relevant more broadly to mechanism design with carryover effects---is that if the mechanism can smooth the payoffs learners experience over time, the resulting dynamics improve even with simple learners. In our model, buffers play this role by allowing a packet to be sent now and processed later. Exploring this idea in other domains is an interesting direction for future work.
In the networking setting studied here, our results leave two central challenges open. The first is to close the gap between the necessary lower bound of 2 and the sufficient factor of 3 in the ratio between the total arrival and service rates required for stability. The second is to extend the analysis to more complex network topologies and general buffer sizes.

\section*{Acknowledgments} 
We want to thank Bruce Hajek for great comments on an earlier version of our paper. \'Eva Tardos was supported in part by 
AFOSR grant FA9550-23-1-0410,  AFOSR grant FA9550-231-0068, and ONR MURI grant N000142412742.  Ariana Abel and Jer\'onimo Mart\'in Duque was supported by the BURE program of the Bowers College of Computing and Information Sciences at Cornell during the summer of 2024.

\bibliographystyle{named}
\bibliography{Queuing-with-a-Buffer-arxiv-V2}

\begin{thebibliography}{}

\bibitem[\protect\citeauthoryear{Auer \bgroup \em et al.\egroup }{2002a}]{DBLP:journals/siamcomp/AuerCFS02}
Peter Auer, Nicol{\`{o}} Cesa{-}Bianchi, Yoav Freund, and Robert~E. Schapire.
\newblock The nonstochastic multiarmed bandit problem.
\newblock {\em {SIAM} J. Comput.}, 32(1):48--77, 2002.

\bibitem[\protect\citeauthoryear{Auer \bgroup \em et al.\egroup }{2002b}]{auer2002nonstochastic}
Peter Auer, Nicolo Cesa-Bianchi, Yoav Freund, and Robert~E Schapire.
\newblock The nonstochastic multiarmed bandit problem.
\newblock {\em SIAM journal on computing}, 32(1):48--77, 2002.

\bibitem[\protect\citeauthoryear{Balcan \bgroup \em et al.\egroup }{2013}]{DBLP:journals/siamcomp/BalcanBM13}
Maria{-}Florina Balcan, Avrim Blum, and Yishay Mansour.
\newblock Circumventing the price of anarchy: Leading dynamics to good behavior.
\newblock {\em {SIAM} J. Comput.}, 42(1):230--264, 2013.

\bibitem[\protect\citeauthoryear{Balseiro and Gur}{2019}]{DBLP:journals/mansci/BalseiroG19}
Santiago~R. Balseiro and Yonatan Gur.
\newblock Learning in repeated auctions with budgets: Regret minimization and equilibrium.
\newblock {\em Manag. Sci.}, 65(9):3952--3968, 2019.

\bibitem[\protect\citeauthoryear{Baudin \bgroup \em et al.\egroup }{2023}]{DBLP:journals/corr/abs-2302-03614}
Lucas Baudin, Marco Scarsini, and Xavier Venel.
\newblock Strategic behavior and no-regret learning in queueing systems.
\newblock {\em CoRR}, abs/2302.03614, 2023.

\bibitem[\protect\citeauthoryear{Busoniu \bgroup \em et al.\egroup }{2008}]{busoniu2008comprehensive}
Lucian Busoniu, Robert Babuska, and Bart De~Schutter.
\newblock A comprehensive survey of multiagent reinforcement learning.
\newblock {\em IEEE Transactions on Systems, Man, and Cybernetics, Part C (Applications and Reviews)}, 38(2):156--172, 2008.

\bibitem[\protect\citeauthoryear{Fikioris and Tardos}{2023}]{fikioris2023liquid}
Giannis Fikioris and {\'E}va Tardos.
\newblock Liquid welfare guarantees for no-regret learning in sequential budgeted auctions.
\newblock In {\em Proceedings of the 24th ACM Conference on Economics and Computation}, pages 678--698, 2023.

\bibitem[\protect\citeauthoryear{Fikioris \bgroup \em et al.\egroup }{2024}]{fikioris2024learning}
Giannis Fikioris, Robert Kleinberg, Yoav Kolumbus, Raunak Kumar, Yishay Mansour, and {\'E}va Tardos.
\newblock Learning in budgeted auctions with spacing objectives.
\newblock {\em arXiv preprint arXiv:2411.04843}, 2024.

\bibitem[\protect\citeauthoryear{Freund \bgroup \em et al.\egroup }{2022}]{DBLP:conf/colt/FreundLW22}
Daniel Freund, Thodoris Lykouris, and Wentao Weng.
\newblock Efficient decentralized multi-agent learning in asymmetric queuing systems.
\newblock In Po{-}Ling Loh and Maxim Raginsky, editors, {\em Conference on Learning Theory, 2-5 July 2022, London, {UK}}, volume 178 of {\em Proceedings of Machine Learning Research}, pages 4080--4084. {PMLR}, 2022.

\bibitem[\protect\citeauthoryear{Freund \bgroup \em et al.\egroup }{2023}]{DBLP:conf/nips/0001LW23-short}
Daniel Freund, Thodoris Lykouris, and Wentao Weng.
\newblock Quantifying the cost of learning in queueing systems.
\newblock In {\em NeurIPS}, 2023.

\bibitem[\protect\citeauthoryear{Fu \bgroup \em et al.\egroup }{2022}]{DBLP:conf/wine/FuHL22}
Hu~Fu, Qun Hu, and Jia'nan Lin.
\newblock Stability of decentralized queueing networks beyond complete bipartite cases.
\newblock In Kristoffer~Arnsfelt Hansen, Tracy~Xiao Liu, and Azarakhsh Malekian, editors, {\em Web and Internet Economics - 18th International Conference, {WINE} 2022, Troy, NY, USA, December 12-15, 2022, Proceedings}, volume 13778 of {\em Lecture Notes in Computer Science}, pages 96--114. Springer, 2022.

\bibitem[\protect\citeauthoryear{Gaitonde and Tardos}{2020}]{DBLP:conf/sigecom/GaitondeT20}
Jason Gaitonde and {\'{E}}va Tardos.
\newblock Stability and learning in strategic queuing systems.
\newblock In P{\'{e}}ter Bir{\'{o}}, Jason~D. Hartline, Michael Ostrovsky, and Ariel~D. Procaccia, editors, {\em {EC} '20: The 21st {ACM} Conference on Economics and Computation, Virtual Event, Hungary, July 13-17, 2020}, pages 319--347. {ACM}, 2020.

\bibitem[\protect\citeauthoryear{Gaitonde and Tardos}{2023}]{DBLP:journals/jacm/GaitondeT23}
Jason Gaitonde and {\'{E}}va Tardos.
\newblock The price of anarchy of strategic queuing systems.
\newblock {\em J. {ACM}}, 70(3):20:1--20:63, 2023.

\bibitem[\protect\citeauthoryear{Gaitonde \bgroup \em et al.\egroup }{2023}]{DBLP:conf/innovations/GaitondeLLLS23}
Jason Gaitonde, Yingkai Li, Bar Light, Brendan Lucier, and Aleksandrs Slivkins.
\newblock Budget pacing in repeated auctions: Regret and efficiency without convergence.
\newblock In Yael~Tauman Kalai, editor, {\em 14th Innovations in Theoretical Computer Science Conference, {ITCS} 2023, January 10-13, 2023, MIT, Cambridge, Massachusetts, {USA}}, volume 251 of {\em LIPIcs}, pages 52:1--52:1. Schloss Dagstuhl - Leibniz-Zentrum f{\"{u}}r Informatik, 2023.

\bibitem[\protect\citeauthoryear{Grosof \bgroup \em et al.\egroup }{2024a}]{DBLP:journals/corr/abs-2405-04102}
Isaac Grosof, Yige Hong, and Mor Harchol{-}Balter.
\newblock Analysis of markovian arrivals and service with applications to intermittent overload.
\newblock {\em CoRR}, abs/2405.04102, 2024.

\bibitem[\protect\citeauthoryear{Grosof \bgroup \em et al.\egroup }{2024b}]{DBLP:journals/sigmetrics/GrosofHHS24}
Isaac Grosof, Yige Hong, Mor Harchol{-}Balter, and Alan Scheller{-}Wolf.
\newblock The {RESET} and {MARC} techniques, with application to multiserver-job analysis.
\newblock {\em {SIGMETRICS} Perform. Evaluation Rev.}, 51(4):6--7, 2024.

\bibitem[\protect\citeauthoryear{Hassin and Haviv}{2003}]{hassin2003queue}
Refael Hassin and Moshe Haviv.
\newblock {\em To Queue or Not to Queue: Equilibrium Behavior in Queueing Systems}.
\newblock International Operations Resea. Springer US, 2003.

\bibitem[\protect\citeauthoryear{Hassin}{2020}]{hassin2020rational}
Refael Hassin.
\newblock {\em Rational Queueing}.
\newblock Chapman and Hall/CRC Series in Operations Research Series. CRC Press, 2020.

\bibitem[\protect\citeauthoryear{Kolumbus and Nisan}{2022a}]{kolumbus2022auctions}
Yoav Kolumbus and Noam Nisan.
\newblock Auctions between regret-minimizing agents.
\newblock In {\em Proceedings of the ACM Web Conference 2022}, pages 100--111, 2022.

\bibitem[\protect\citeauthoryear{Kolumbus and Nisan}{2022b}]{kolumbus2022and}
Yoav Kolumbus and Noam Nisan.
\newblock How and why to manipulate your own agent: On the incentives of users of learning agents.
\newblock {\em Advances in Neural Information Processing Systems}, 35:28080--28094, 2022.

\bibitem[\protect\citeauthoryear{Koutsoupias and Papadimitriou}{1999}]{DBLP:conf/stacs/KoutsoupiasP99}
Elias Koutsoupias and Christos~H. Papadimitriou.
\newblock Worst-case equilibria.
\newblock In {\em {STACS} 99, 16th Annual Symposium on Theoretical Aspects of Computer Science}, pages 404--413. Springer, 1999.

\bibitem[\protect\citeauthoryear{Krishnasamy \bgroup \em et al.\egroup }{2016}]{DBLP:conf/nips/KrishnasamySJS16}
Subhashini Krishnasamy, Rajat Sen, Ramesh Johari, and Sanjay Shakkottai.
\newblock Regret of queueing bandits.
\newblock In {\em Advances in Neural Information Processing Systems 29: Annual Conference on Neural Information Processing Systems 2016, December 5-10, 2016, Barcelona, Spain}, pages 1669--1677, 2016.

\bibitem[\protect\citeauthoryear{Littman}{1994}]{littman1994markov}
Michael~L Littman.
\newblock Markov games as a framework for multi-agent reinforcement learning.
\newblock In {\em Machine learning proceedings 1994}, pages 157--163. Elsevier, 1994.

\bibitem[\protect\citeauthoryear{Pemantle and Rosenthal}{1999}]{pemantle1999moment}
Robin Pemantle and Jeffrey~S Rosenthal.
\newblock Moment conditions for a sequence with negative drift to be uniformly bounded in ${L}^r$.
\newblock {\em Stochastic Processes and their Applications}, 82(1):143--155, 1999.

\bibitem[\protect\citeauthoryear{Roughgarden}{2015}]{DBLP:journals/jacm/Roughgarden15}
Tim Roughgarden.
\newblock Intrinsic robustness of the price of anarchy.
\newblock {\em J. {ACM}}, 62(5):32:1--32:42, 2015.

\bibitem[\protect\citeauthoryear{Sentenac \bgroup \em et al.\egroup }{2021}]{DBLP:conf/nips/SentenacBP21-short}
Flore Sentenac, Etienne Boursier, and Vianney Perchet.
\newblock Decentralized learning in online queuing systems.
\newblock In {\em NeurIPS}, 2021.

\bibitem[\protect\citeauthoryear{Shortle \bgroup \em et al.\egroup }{2018}]{queuing_theory}
John~F Shortle, James~M Thompson, Donald Gross, and Carl~M Harris.
\newblock {\em Fundamentals of Queueing Theory}.
\newblock Wiley, 2018.

\bibitem[\protect\citeauthoryear{Slivkins}{2019}]{slivkins2019introduction}
Aleksandrs Slivkins.
\newblock Introduction to multi-armed bandits.
\newblock {\em Found. Trends Mach. Learn.}, 12(1-2):1--286, 2019.

\bibitem[\protect\citeauthoryear{Syrgkanis and Tardos}{2013}]{DBLP:conf/stoc/SyrgkanisT13}
Vasilis Syrgkanis and {\'{E}}va Tardos.
\newblock Composable and efficient mechanisms.
\newblock In Dan Boneh, Tim Roughgarden, and Joan Feigenbaum, editors, {\em Symposium on Theory of Computing Conference, STOC'13, Palo Alto, CA, USA, June 1-4, 2013}, pages 211--220. {ACM}, 2013.

\bibitem[\protect\citeauthoryear{Walton and Xu}{2021}]{doi:10.1287/educ.2021.0235}
Neil Walton and Kuang Xu.
\newblock {\em Learning and Information in Stochastic Networks and Queues}, chapter~6, pages 161--198.
\newblock 2021.

\bibitem[\protect\citeauthoryear{Zhang \bgroup \em et al.\egroup }{2021}]{zhang2021multi}
Kaiqing Zhang, Zhuoran Yang, and Tamer Ba{\c{s}}ar.
\newblock Multi-agent reinforcement learning: A selective overview of theories and algorithms.
\newblock {\em Handbook of reinforcement learning and control}, pages 321--384, 2021.

\end{thebibliography}

\end{document}